\documentclass[11pt]{article}
\usepackage[utf8]{inputenc}
\usepackage{fullpage}
\usepackage{todonotes}
\usepackage{amsmath}
\usepackage{amsthm}
\usepackage{amssymb}
\usepackage{hyperref}
\usepackage{comment}
\usepackage[capitalise]{cleveref}
\usepackage[nocompress]{cite}
\usepackage{xspace}
\usepackage[shortlabels]{enumitem}
\usepackage{thmtools}
\usepackage{thm-restate}

\theoremstyle{plain}
\newtheorem{theorem}{Theorem}
\newtheorem{lemma}{Lemma}  
\newtheorem{corollary}{Corollary} 
  
\newtheorem{fact}{Fact}

\theoremstyle{definition}
\newtheorem{definition}{Definition}

\usepackage[linesnumbered,lined,boxed,ruled,vlined]{algorithm2e}

\newcommand{\disth}{\mathrm{dist}_H}
\newcommand{\Lap}{\mathrm{Lap}}
\SetKwInput{KwInput}{Input}
\SetKwInput{KwOutput}{Output}
\newcommand{\thresh}{\mathrm{Thresh}}
\newcommand{\rev}[1]{#1^{\mathrm{rev}}}
\newcommand{\alg}{\mathrm{Alg}}
\newcommand{\Out}{\mathrm{Out}}
\newcommand{\range}{\mathrm{range}}

\title{Differentially Private Approximate Pattern Matching}
\date{}

\begin{document}
\author{Teresa Anna Steiner\footnote{DTU Compute, Technical University of Denmark, Kongens Lyngby, Denmark. Email:\url{terst@dtu.dk}}}%\footnote{This work was supported by a research grant (VIL51463) from VILLUM FONDEN.}}
\maketitle
\begin{abstract}
Differential privacy is the de-facto privacy standard in data analysis and widely researched in various application areas. On the other hand, analyzing sequences, or \emph{strings}, is essential to many modern data analysis tasks, and those data often include highly sensitive personal data. While the problem of
sanitizing sequential data to protect privacy has received growing attention, there is a surprising lack of theoretical studies of algorithms analyzing sequential data that preserve differential privacy while giving \emph{provable guarantees on the accuracy} of such an algorithm. The goal of this paper is to initiate such a study.

Specifically, in this paper, we consider the $k$-approximate pattern matching problem under differential privacy, where the goal is to report or count all substrings of a given string $S$ which have a Hamming distance at most $k$ to a pattern $P$, or decide whether such a substring exists. In our definition of privacy, \emph{individual positions of the string $S$} are protected. To be able to answer queries under differential privacy, we allow some slack on $k$, i.e. we allow reporting or counting substrings of $S$ with a distance at most $(1+\gamma)k+\alpha$ to $P$, for a multiplicative error $\gamma$ and an additive error $\alpha$. We analyze which values of $\alpha$ and $\gamma$ are necessary or sufficient to solve the $k$-approximate pattern matching problem while satisfying $\epsilon$-differential privacy. Let $n$ denote the length of $S$. We give
\begin{itemize}
    \item an $\epsilon$-differentially private algorithm with an additive error of $O(\epsilon^{-1}\log n)$ and no multiplicative error for the existence variant;
    \item an $\epsilon$-differentially private algorithm with an additive error $O(\epsilon^{-1}\max(k,\log n)\cdot\log n)$ for the counting variant;
    \item an $\epsilon$-differentially private algorithm with an additive error of $O(\epsilon^{-1}\log n)$ and multiplicative error $O(1)$ for the reporting variant for a special class of patterns.
\end{itemize}
The error bounds hold with high probability.
All of these algorithms return a witness, that is, if there exists a substring of $S$ with distance at most $k$ to $P$, then the algorithm returns a substring of $S$ with distance at most $(1+\gamma)k+\alpha$ to $P$.

Further, we complement these results by a lower bound, showing that any algorithm for the existence variant which also returns a witness must have an additive error of $\Omega(\epsilon^{-1}\log n)$ with constant probability.
\end{abstract}

\section{Introduction}

%\begin{itemize}
%    \item Write a paragraph about what differential privacy is and why it is important
%    \item write a paragraph about what pattern matching is and why it is important
%    \item write a paragraph about why privacy protecting pattern matching algorithms should be important: the data is sensitive; list previous approaches here!
%\end{itemize}

Analyzing sequential data is essential to many modern data analysis tasks, including signal processing, route planning, and genetic matching. Since those data can include highly sensitive personal data, the problem of
\emph{sanitizing sequential data to protect privacy} while \emph{preserving patterns that occur within these sequences} has received growing attention \cite{DBLP:journals/isci/ChenFMDW13,DBLP:conf/ccs/ChenAC12,DBLP:journals/tetc/WeiLYZL21,DBLP:conf/mdm/MaruseacG16,DBLP:conf/cikm/BonomiX13,DBLP:conf/ifip12/AjalaAIL18,DBLP:journals/tkdd/BernardiniCCGLP21,DBLP:journals/tkde/BernardiniCGGLPPPSS23,DBLP:journals/kbs/KomishaniAD16,DBLP:conf/nips/KimGKY21,DBLP:conf/cikm/BonomiXCF12,DBLP:conf/kdd/BhaskarLST10,DBLP:conf/kdd/ChenFDS12,DBLP:journals/pvldb/HeCMPS15,DBLP:conf/sigmod/ZhangXX16,DBLP:conf/dasfaa/LiWYCYL18,DBLP:conf/infocom/WangLPRLC20}. The applications considered in these papers range from genetic matching \cite{DBLP:journals/tetc/WeiLYZL21} over natural language processing \cite{DBLP:conf/nips/KimGKY21,DBLP:conf/ccs/ChenAC12} to travel pattern mining \cite{DBLP:conf/mdm/MaruseacG16,DBLP:journals/isci/ChenFMDW13,DBLP:journals/kbs/KomishaniAD16,DBLP:conf/kdd/ChenFDS12,DBLP:journals/pvldb/HeCMPS15,DBLP:conf/sigmod/ZhangXX16}. %They can be categorized in differentially private schemes for finding frequent patterns or $q$-grams~\cite{}, publishing synthetic data preserving patterns~\cite{}, and combinatorial string sensitisation~\cite{}. \textcolor{red}{not sure which formulation is best here. maybe move previous sentence to related work and go into more detail there.}
 These works partially use differential privacy \cite{DBLP:journals/tetc/WeiLYZL21,DBLP:conf/mdm/MaruseacG16,DBLP:conf/cikm/BonomiX13,DBLP:conf/ccs/ChenAC12,DBLP:conf/nips/KimGKY21,DBLP:conf/cikm/BonomiXCF12,DBLP:conf/kdd/BhaskarLST10,DBLP:conf/kdd/ChenFDS12,DBLP:journals/pvldb/HeCMPS15,DBLP:conf/sigmod/ZhangXX16,DBLP:conf/dasfaa/LiWYCYL18,DBLP:conf/infocom/WangLPRLC20} or other privacy measures \cite{DBLP:conf/ifip12/AjalaAIL18,DBLP:journals/tkdd/BernardiniCCGLP21,DBLP:journals/tkde/BernardiniCGGLPPPSS23,DBLP:journals/isci/ChenFMDW13,DBLP:journals/kbs/KomishaniAD16}. 
The utilities of the proposed algorithms are shown by extensive experiments. Despite this effort led by practitioners, there is a lack of theoretical studies of algorithms analyzing sequential data that preserve differential privacy while giving \emph{provable guarantees on the accuracy} of such an algorithm. The goal of this paper is to initiate such a study.

Differential privacy is the de-facto privacy standard used in modern data analysis~\cite{DBLP:journals/tbd/YangWRY21}. Its definition offers strong privacy guarantees and is due to Dwork~et~al.~\cite{DBLP:conf/tcc/DworkMNS06}. Informally, the definition states that the output distributions of an algorithm should be close on close data sets, i.e., the output should not depend much on any single data point. In more detail, we call two data sets which differ in a single data point \emph{neighbouring}. A randomized algorithm is \emph{$\epsilon$-differentially private}, if for any two neighbouring input data sets, the output distributions of the algorithm differ by at most a factor of $e^{\epsilon}$. %\textcolor{red}{maybe add some citation / text about the importance of differential privacy here}

A natural data type to model sequential data is a \emph{string}, which is a sequence of symbols drawn from some predefined alphabet. Strings are used to model any type of text data, as well as genetic data and event series. One of the most fundamental problems in string algorithms is the pattern matching problem: For a string $S$ and a pattern string $P$, decide if $P$ occurs in $S$ (\emph{existence}), count the occurrences of $P$ in $S$ (\emph{counting}), or report all positions in $S$ where $P$ occurs (\emph{reporting}). The pattern matching problem and its variants have been an active research field for more than 50 years with applications ranching from signal processing over computational biology to information retrieval. %\textcolor{red}{add some citations as to why this is important})
%The goal of this paper is to initiate such a study. Specifically, we study the approximate pattern matching problem and show that combining well-known techniques from differential privacy \cite{} with modern techniques used by the pattern matching community to solve the approximate pattern matching problem \cite{} can be used to give interesting theoretical upper and lower bounds on the error needed by any differential private algorithm solving the approximate pattern problem.

%In the following, we describe the problem and our results in more detail. 

In this work, we begin a theoretical study of \emph{differentially private pattern matching} for strings. Specifically, we study the approximate pattern matching problem and show that combining well-known techniques from differential privacy \cite{DBLP:journals/fttcs/DworkR14} with modern techniques used by the pattern matching community to solve the approximate pattern matching problem \cite{conf/focs/Charalampopoulos20} can be used to prove interesting new theoretical upper and lower bounds on the error needed by any differentially private algorithm solving the approximate pattern matching problem. We see this as a proof of concept that the field of \emph{differentially private string algorithms} is a promising direction for future research.

In the following, we describe the problem considered in this work in more detail.
\paragraph*{Privacy model and motivation.}
In this paper we focus on protecting \emph{individual positions in the string $S$}, that is, the pattern matching algorithm should have similar output distributions when matching $P$ in $S$ and $T$, if $S$ and $T$ differ in few positions. That is, we call two strings $S$ and $S'$ neighbouring, if they differ in one position. This privacy model has also been used for strings by Fichtenberger~et~al.~\cite{jalajmonika} for the problem of counting all occurrences of any pattern of a given length in a stream. It corresponds to \emph{event-level privacy} for continual observation, i.e., instead of protecting an entire user's data, single events are protected \cite{DBLP:conf/stoc/DworkNPR10}. Since the output of the algorithm has a similar distribution whether any single event happened or not, this can be seen as providing \emph{plausible deniability} of any given event. Thus, this model makes sense in settings where a user cares about single events or outliers in their behaviour being concealed, while still allowing the service to draw conclusions about their general behaviour. For example, the string could be a sequence of locations a person visited, and hiding any single position in that sequence corresponds to hiding whether a person visited any particular location at a given time or not. For another example, the string can be a list of items bought by a customer through an online service, and any single purchase is masked. This definition can still allow trends to be detected, e.g. if a user buys chocolate every day, a differentially private algorithm may reveal that the person buys lots of chocolate; however, if a user buys a single sensitive item, e.g. a pregnancy test, this data is concealed. %For analyzing DNA data, it is known that there are diseases for which the mutation of a \emph{single} gene strongly correlates with the development of the disease%, which can then be used to link genes to e.g. patient data from a hospital 
%\cite{DNAidentifiability}.% In all of these examples, concealing few positions in the string can have a huge positive impact on personal privacy.

%There is a line of work on string sanitization focusing on hiding a \emph{given set of sensitive patterns} by replacing a letters in the text \cite{DBLP:conf/ifip12/AjalaAIL18,DBLP:journals/tkdd/BernardiniCCGLP21,DBLP:journals/tkde/BernardiniCGGLPPPSS23}. 
%\textcolor{red}{(write for which applications this setup makes sense)}

\paragraph*{Approximate pattern matching.}
Note that we cannot hope to solve the pattern matching problem \emph{exactly} while satisfying this definition of differential privacy: For any pattern $P$, we can easily find strings $S$ and $S'$ such that $P$ occurs in $S$, $P$ does not occur in $S'$, and $S$ and $S'$ differ in only one position. Thus, any reasonable solution to the exact pattern matching problem with pattern $P$ should be able to differentiate between $S$ and $S'$, which contradicts the goal of differential privacy. 

Therefore we study the \emph{$k$-approximate pattern matching problem}: For a pattern $P$ of length $m$, a string $S$ of length $n\geq m$, and a parameter $k\leq m$, we want to find all substrings of length $m$ of $S$, such that the distance between the substring and $P$ is at most $k$. This problem has been extensively studied in the non-private setting (recent work includes \cite{DBLP:conf/cpm/Starikovskaya17,DBLP:journals/ipl/ZhangA17,DBLP:conf/icalp/GawrychowskiU18,conf/focs/Charalampopoulos20,DBLP:journals/tcs/BernardiniPPR20}, see also the survey by Navarro~\cite{DBLP:journals/csur/Navarro01}) since it captures several applications more fully than exact matching: In many applications, the string and the pattern might suffer some corruption, e.g. mutation in DNA sequences, measurement or transmission errors, or typing errors \cite{DBLP:journals/csur/Navarro01}. In this work, we consider the Hamming distance as distance measure. In order to design algorithms that fulfill differential privacy, we allow some slack on $k$: We want to find all length-$m$ substrings (given by their starting and ending position in $S$) of distance at most $k$ to $P$, but we allow the algorithm to return length-$m$ substrings of distance at most $(1+\gamma)k+\alpha$, for a \emph{multiplicative error} $\gamma$ and an \emph{additive error} $\alpha$. We also consider the natural \emph{counting} and \emph{existence} variants of this problem (the formal definitions of these problems are given in Section \ref{sec:prelims}). The goal is to analyze which values of $\gamma$ and $\alpha$ are possible and necessary to solve the approximate pattern matching problem while preserving $\epsilon$-differential privacy. 

\paragraph*{Results.}
First, we note that there is a trivial algorithm with additive error $O(m)$, which is $\epsilon$-differentially private for all $\epsilon$: We simply output all substrings of $S$, i.e. all pairs $(i,i+m-1)$ for $i\leq n-m$. Since this is independent of the string $S$, the algorithm is differentially private by default, and since the true distance is always a value between $0$ and $m$, the additive error is at most~$m$. 

In this paper, we give new trade-offs for the existence, counting and reporting variants of the problem. First, we give an algorithm for the existence variant achieving $O(\log n)$ additive error and no multiplicative error. Then, for counting and reporting, we use results on (non-private) approximate pattern matching \cite{conf/focs/Charalampopoulos20} to differentiate between patterns fulfilling different properties: If the pattern is close to a periodic string with a small enough period, we can exploit that to give an algorithm for the reporting variant of the approximate pattern matching problem with constant multiplicative error and $O(\log n)$ additive error. Otherwise, we can use the results in \cite{conf/focs/Charalampopoulos20} to bound the number of substrings in $S$ which can be close to $P$, and use that fact to give an algorithm for the counting variant. Our upper bound results are summarized in the following two theorems.
\begin{theorem}[Summary of Lemma \ref{cor:existence}, Theorem \ref{thm:periodic}, Theorem \ref{thm:non-periodic}, and Lemma \ref{lem:smallk}]\label{thm:main}
    Let $n$ denote the length of input string $S$, $m\leq n$ the length of pattern $P$, and $k\leq m$ an integer. 
    \begin{enumerate}
        \item There exists an algorithm for the existence variant of the $k$-approximate pattern matching problem which with probability $1-\beta$ has an additive error of at most $\alpha=O(\log(n/\beta)/\epsilon)$ and a multiplicative error $\gamma=0$.
        \item For $k=\Omega(\log(n/\beta)/\epsilon)$, there exists an algorithm for the counting variant of the $k$-approximate pattern matching problem which with probability $1-\beta$ has a multiplicative error of at most $\gamma=O(\log(n/\beta)/\epsilon)$ and an additive error $\alpha=0$.
        \item For $k=O(\log(n/\beta)/\epsilon)$, there exists an algorithm for the counting variant of the $k$-approximate pattern matching problem which with probability $1-\beta$ has an additive error of at most $\alpha=O(\log^2(n/\beta)/\epsilon^2)$ and a multiplicative error $\gamma=0$.
    \end{enumerate}
    Further, all of these algorithms return a \emph{witness}, i.e., a length-$m$ substring of $S$ with Hamming distance at most $(1+\gamma)k+\alpha$ to $P$.
\end{theorem}

\begin{theorem}[informal version of Theorem \ref{thm:periodic}]\label{thm:periodic_int}
 Let $P$ be a string of length $m$. If $P$ has Hamming distance at most $2k$ from a periodic string of period at most $qm/((\log(n/\beta)/\epsilon)+k)$, for some suitable constant $q$, then there exists an algorithm for the reporting variant of the $k$-approximate pattern matching problem for pattern $P$ and any string $S$ of length $n$ which with probability $1-\beta$ has a multiplicative error of $O(1)$ and an additive error of $O(\log(n/\beta)/\epsilon)$.
\end{theorem}

We complement these results with lower bounds on the necessary additive error for the $k$-approximate pattern matching problem under $\epsilon$-differential privacy. These lower bounds specifically show that the additive error for the existence variant from Theorem~\ref{thm:main} % and the additive error from Theorem~\ref{thm:periodic_int} 
is asymptotically optimal for~$m\ll n$:

\begin{theorem}[Informal version of Theorem \ref{thm:lowerbound}]\label{thm:lowerbound_int}
    Let $P$ be any string of length $m$ and let $k<m$ be an integer. Assume there is an $\epsilon$-differentially private algorithm which solves the existence variant of the $k$-approximate pattern matching problem for pattern $P$ and any string $S$ and returns a witness, with an additive error at most $\alpha$ with constant probability. Then either $\alpha=\Omega(m-k)$, or both $m=\Omega(\epsilon^{-1}\log n)$ and $\alpha=\Omega(\epsilon^{-1}\log(n/m))$.
\end{theorem}
Note that Theorem~\ref{thm:lowerbound_int} gives a lower bound that holds for \emph{any} pattern $P$, no matter if it is close to a periodic substring of small period, or not.

In this work, we mostly care about the privacy-to-accuracy trade-off of the problem. However, for completeness, we show in Appendix \ref{sec:runtime}, that the algorithms achieving the upper bounds stated above run in time $O(nm+m^3)$, assuming that any needed random noise can be drawn in constant time. We did not try to optimize this run time. 

\paragraph*{Related work.}
Fichtenberger~et~al.~\cite{jalajmonika} show how to count all patterns of a bounded length over a stream while preserving differential privacy. It is given as a direct application of their general differentially private counting algorithm. Their privacy model is the same as ours, however, their error definition is an error on the \emph{value of the count}, instead of an error on the Hamming distance, as in our paper.

There is a large body of work on mining frequent patterns or $q$-grams (substrings of length $q$) from a set of strings while satisfying differential privacy \cite{DBLP:conf/kdd/BhaskarLST10,DBLP:conf/ccs/ChenAC12,DBLP:conf/cikm/BonomiX13,DBLP:conf/mdm/MaruseacG16,DBLP:conf/cikm/BonomiXCF12,DBLP:conf/nips/KimGKY21,DBLP:conf/kdd/ChenFDS12,DBLP:conf/sigmod/ZhangXX16,DBLP:conf/dasfaa/LiWYCYL18}. In those works, the input data set consists of multiple strings, and two neighbouring data sets differ in one string in the set. The utilities of these algorithm are evaluated by experiments.
    %\item publishing sanitized data sets.
    
    There is a line of work on \emph{combinatorial string sanitization} focusing on hiding a \emph{given set of sensitive patterns} \cite{DBLP:conf/ifip12/AjalaAIL18,DBLP:journals/tkdd/BernardiniCCGLP21,DBLP:journals/tkde/BernardiniCGGLPPPSS23}. Ajala~et~al.~\cite{DBLP:conf/ifip12/AjalaAIL18} consider sanitizing the string by replacing letters. They show that the problem of finding the minimum number of letters to be replaced is NP-hard and propose an algorithm. Bernardini~et~al.~\cite{DBLP:journals/tkdd/BernardiniCCGLP21} propose an algorithm for finding the minimal length string maintaining the order and frequency of all non-sensitive patterns, and another algorithm for finding a string maintaining the order and frequency of all non-sensitive patterns while minimizing the edit distance between the original string and the output string. Bernadini~et~al.~\cite{DBLP:journals/tkde/BernardiniCGGLPPPSS23} study the connection between string sanitization and frequent pattern mining. Compared to our work, they mask \emph{all} occurrences of sensitive patterns, however, the specific patterns have to be given in advance. On the other hand, our definition hides any single (or any set of few) occurrences of \emph{any} potentially sensitive pattern. Note that in those works, the goal is to mask \emph{exact} occurrences of the sensitive patterns, i.e. it still allows occurrences of substrings which are close to a sensitive pattern.
    
    There is previous work on private pattern matching from a cryptographic perspective with applications in genetic matching \cite{mahdi2021privacy,DBLP:conf/acsac/MainardiBP19,DBLP:conf/ccs/Troncoso-PastorizaKC07,DBLP:journals/access/QinZZX20,DBLP:journals/bioinformatics/ShimizuNR16,DBLP:journals/tcbb/SudoJNS19,DBLP:journals/soco/WeiZX18}: In the model considered in these works, data is held by one party (or the cloud) and queries are sent by another (or multiple other) parties; encryption is used to ensure privacy of the data and the query. In these works, the query party can find out whether their query pattern occurs in the string or collection of strings in the data, while nothing else about the data is revealed to the query party and the query is not revealed to the data holder. In a similar model, two parties each hold a string and want to compare how similar they are, without revealing anything else to each other~\cite{DBLP:journals/ijdsa/VaiwsriRC22}. %In \cite{mahdi2021privacy}, the data is a set of strings, and they consider pattern matching queries. In this works, the goal is to give exact answers to queries while protecting the data from any unauthorized party. 
    Note that the goal in differential privacy is orthogonal to these privacy definitions: In our definition, the data holder knows everything; however, the \emph{query answer} should conceal any individual string positions of the data holder's string.

    \paragraph*{Paper organization.} The rest of the paper is organized as follows. In Section~\ref{sec:prelims}, we formally define the problem and recall some definitions and theorems for differential privacy and strings. In Section~\ref{sec:upperbounds}, we prove Theorems~\ref{thm:main} and \ref{thm:periodic_int}. In Section~\ref{sec:lowerbound}, we prove Theorem~\ref{thm:lowerbound_int}. Finally, we conclude with some directions for future research~(Section~\ref{sec:conclusion}). In Appendix~\ref{sec:runtime}, we analyze the runtime of our algorithms.
%First, we note that a simple algorithm can yield a multiplicative error of $0$ and an additive error of $O(m/\epsilon)$ with constant probability: As a first step we compute, for every length-$m$ substring of $S$, the true Hamming distance between that substring and $P$. This gives an output vector containing $n-m$ distance values. Now, note that if one position of $S$ was different, than at most $m$ of these distance values can change by at most 1. Then the so-called \emph{Laplace mechanism} \cite{} yields the claimed error bound.

%Trivial upper bound: for every position $i$ in the text, compute the true distance $d_i=\disth(P,S[i,i+m-1])$ and output $d_i+\Lap(m/\epsilon)$. Since changing one position in the string can change at most $m$ values of $d_i$ by at most 1, this is $\epsilon$-differentially private. Reporting, counting and existence can be solved as post-processing. This gives a multiplicative error $\gamma=0$ and an additive error $\alpha$ with $\alpha=O(m/\epsilon) \ln(1/\beta)$ with probability $1-\beta$.

\section{Preliminaries}\label{sec:prelims}
We denote an interval of integers $\{a,a+1,\dots, b\}$ as $[a,b]$. 
\subsection{String Preliminaries}
A \emph{string} $S$ of length $n$ is a sequence $S[0]S[1]\dots S[n-1]$ of symbols from an alphabet $\Sigma$. The length of $S$ is denoted $|S|$. We call $S[a,b]:=S[a]S[a+1]\dots S[b]$ a \emph{substring} of $S$. We denote by $\rev{S}:=S[n-1]S[n-2]\dots S[0]$ the \emph{reverse} of string $S$. 
For $k\in\mathbb{N}\cup \{\infty\}$ we denote by $S^k$ the string obtained by concatenating $S$ $k$ times. A string $S$ is called primitive if there does not exist a string $T$ such that $S=T^k$ for $k\geq 2$.

A \emph{period} of a string $S$ is a number $\pi\in[0,n-1]$ such that $S[i]=S[i+\pi]$ for all $i\in[0,n-1-\pi]$. A string $S$ is \emph{periodic} if it has a period $\pi$ with $\pi<n/2$. 

The \emph{Hamming distance} between two strings $S$ and $T$ with $n=|T|=|S|$ is defined as 
\begin{align*}
    \disth(T,S)=|\{i\in[0,n-1]: T[i]\neq S[i]\}|.
\end{align*}
For a string $P$ of length $m$ and a string $S$ of length $n$ with $n\geq m$, $i\in[0,n-m]$, we call $S[i,i+m-1]$ a \emph{$k$-mismatch occurrence} if $\disth(S[i,i+m-1],P)\leq k$.
%\item define $k$-mismatch and add to problem definitions where it makes sense
%\item define primitive string, periodicity, $\rev{T}$

\subsection{Privacy Definition and Problem Definitions}

Two strings $S$ and $S'$ of length $n$ are defined as \emph{neighbouring}, if their Hamming distance is one, i.e., if they differ in one position.

We generally define a \emph{pattern matching algorithm} to be an algorithm taking as input a string $S$ of length $n$ and a pattern $P$, and outputting either a Boolean value (\emph{existence}), a natural number in $[0,n-1]$ (\emph{counting}), or a subset of $[0,n-1]$ (\emph{reporting}). 

We say a pattern matching algorithm $\alg:\Sigma^{*}\times\Sigma^m\rightarrow \mathrm{range}(\alg)$ is \emph{$\epsilon-$differentially private}, if for all $\mathrm{Out}\subseteq \mathrm{range}(\alg)$, all patterns $P$ of length $m$ and all pairs of neighbouring strings $S$ and $S'$,
\begin{align*}
\Pr(\alg(S,P)\in \mathrm{Out})\leq e^{\epsilon}\cdot\Pr(\alg(S',P)\in \mathrm{Out}),
\end{align*}
where the probabilities are taken over the internal randomness of $\alg$.
%\item \textcolor{red}{instead: approximate $k$-mismatch problem, as in \cite{DBLP:conf/soda/CliffordFPSS16}?}
\begin{definition}[$k$-approximate pattern matching problem with one-sided error, reporting variant]% as follows: 
Given a string $S$ of length $n$, a pattern $P$ of length $m$ and a parameter $k$, output a set of indices $I\in[0,n-m]$ such that
\begin{enumerate}
 %\textcolor{red}{(maybe: $k(1+\alpha)$)}
\item If $\disth(P,S[i,i+m-1])\leq k$ %-\alpha$ 
for an $i\in[0,n-m]$, then $i\in I$,
\item If $i\in I$ then $\disth(P,S[i,i+m-1])\leq (1+\gamma) k+\alpha$.%\textcolor{red}{(maybe: $k(1-\alpha)$)}
\end{enumerate}
We call $\gamma$ the \emph{multiplicative error} and $\alpha$ the \emph{additive error}.
\end{definition}

In the following, let $c_x(S,P)$ denote the number of positions $i$ in $S$  such that $\disth(P,S[i,i+m-1])\leq x$. If $P$ is clear from context, we will sometimes write $c_x(S)$ for $c_x(S,P)$.
\begin{definition}[$k$-approximate pattern matching problem with one-sided error, counting variant]
%We define the \emph{$k$-approximate pattern matching problem with one-sided error (counting)} as follows: 
Given a string $S$ of length $n$, a pattern $P$ of length $m$ and a parameter $k$, output a number $c$ such that
\begin{enumerate}
\item $c\geq c_{k}(S,P)$, 
\item  $c\leq c_{(1+\gamma)k+\alpha}(S,P)$.
\end{enumerate}
Further, if $c>0$, additionally output a position $i$ fulfilling $\disth(P,S[i,i+m-1])\leq (1+\gamma)k+\alpha$. We call $i$ a \emph{witness}. 
%\textcolor{red}{(Again, maybe: $k(1+-\alpha)$)}
We call $\gamma$ the \emph{multiplicative error} and $\alpha$ the \emph{additive error}.
\end{definition}

\begin{definition}[$k$-approximate pattern matching problem with one-sided error, existence variant]
    Given a string $S$ of length $n$, a pattern $P$ of length $m$ and a parameter  $k$, output
\begin{enumerate}
\item YES, if there exists $i\in[0,n-m]$ such that $\disth(P,S[i,i+m-1])\leq k$, 
\item  NO, if there does not exist $i\in[0,n-m]$ such that $\disth(P,S[i,i+m-1])\leq (1+\gamma)k+\alpha$.
\end{enumerate}
Further, if the answer is YES, additionally output a position $i$ fulfilling $\disth(P,S[i,i+m-1])\leq (1+\gamma)k+\alpha$.
%\textcolor{red}{(Again, maybe: $k(1+-\alpha)$)}
We call $i$ a \emph{witness}. We call $\gamma$ the \emph{multiplicative error} and $\alpha$ the \emph{additive error}.
\end{definition}
%We define the \emph{$k$-approximate pattern matching problem with one-sided error (existence)} as follows: 

%\textcolor{red}{Or: just consider the actual additive error between $c$ and $c_k$}

\subsection{Privacy Preliminaries}
First, we collect some definitions to introduce the \emph{Laplace mechanism}.
   \begin{definition}[$L_1$-sensitivity]\label{def:sensitivity}
     Let $f$ be a function $f:\chi\rightarrow \mathbb{R}^k$ for some universe $\chi$. The \emph{$L_1$-sensitivity of $f$} is defined as
    \begin{align}
    \max_{x,y\textnormal{ neighboring}}||f(x)-f(y)||_{1}.\label{eq:sensitivity}
    \end{align} %\SR{is k=1 or p=1? If it's k then we still need to specify p}
    \end{definition}
\begin{definition}\label{def:laplace}
The \emph{Laplace distribution} centered at $0$ with scale $b$ is the distribution with probability density function %\SR{whats beta?}
\begin{align*}
    f_{\Lap(b)}(x)=\frac{1}{2b}\exp\left(\frac{-|x|}{b}\right).
\end{align*}
 We use $X\sim \Lap(b)$ or just $\Lap(b)$ to denote a random variable $X$ distributed according to $f_{\Lap(b)}(x)$.
\end{definition}
\begin{lemma}[Theorem 3.6 in \cite{DBLP:journals/fttcs/DworkR14}: Laplace Mechanism] \label{lem:Laplacemech} Let $f$ be any function $f:\chi\rightarrow \mathbb{R}^k$ with $L_1$-sensitivity $\Delta_1$. Let $Y_i\sim \Lap(\Delta_1/\epsilon)$ for $i\in[k]$. The mechanism defined as:
\begin{align*}
    A(x)=f(x)+(Y_1,\dots,Y_k)
\end{align*}
satisfies $\epsilon$-differential privacy.
\end{lemma}

The following fact follows directly from the definition of differential privacy, and extends the privacy definition from neighbouring input strings to inputs which have small distance from each other.
\begin{lemma}[Group Privacy for Pattern Matching]\label{fact:group_privacy}
Let $S$ and $S'$ have a Hamming distance at most $\ell$, i.e. $\disth(S,S')\leq \ell$. Let $\alg$ be an $\epsilon$-differentially private pattern matching algorithm. Then for any pattern $P$,
\begin{align*}
    \Pr(\alg(S,P)\in \Out)\leq e^{\ell\epsilon} \cdot \Pr(\alg(S',P)\in\Out).
\end{align*}
\end{lemma}
The following is a well-known Fact which follows immediately from the definition of differential privacy.
\begin{lemma}[Composition Theorem]\label{fact:composition_theorem} Let $\alg_1:\chi\rightarrow\mathrm{range}(\alg_1)$ be an $\epsilon_1$-differentially private algorithm and $\alg_2:\chi\times \mathrm{range}(\alg_1)\rightarrow\mathrm{range}(\alg_2)$ be an an $\epsilon_2$-differentially private algorithm. Then $(\alg_1,\alg_2 \circ \alg_1):\chi\rightarrow\mathrm{range}(\alg_1)\times\mathrm{range}(\alg_2)$ is $(\epsilon_1+\epsilon_2)$-differentially private. \end{lemma}

The following Lemma is a variant of \emph{parallel composition}~\cite{DBLP:journals/cacm/McSherry10} of differential privacy, applied to strings. It says that if we run independent $\epsilon$-differentially private algorithms on disjoint substrings, then the resulting algorithm is still $\epsilon$-differentially private: 
\begin{lemma}\label{lem:composing_disjiont}
    Let $\alg_1$ and $\alg_2$ be independent $\epsilon$-differentially private pattern matching algorithms and let $S$ be a string. Further, let $[a,b]\subseteq[0,n-1]$ and $[c,d]\subseteq[0,n-1]$ and $[a,b]\cap[c,d]=\emptyset$. Then  algorithm $\alg_3(S):=(\alg_1(S[a,b]),\alg_2(S[c,d]))$ is $\epsilon$-differentially private.
\end{lemma}
\begin{proof}
    Let $S$ and $S'$ be neighbouring strings and let $P$ be a pattern. Let $i$ be the position where $S[i]\neq S'[i]$. Let $\Out=(\Out_1,\Out_2)\subseteq\mathrm{range}(\alg_1)\times\mathrm{range}(\alg_2)=\range(\alg_3)$. If $i\in[a,b]$, then
    \begin{align*}
        \Pr(\alg_3(S,P)\in\Out)&=\Pr((\alg_1(S[a,b],P),\alg_2(S[c,d],P))\in(\Out_1,\Out_2))\\&=\Pr(\alg_1(S[a,b],P)\in \Out_1)\cdot\Pr(\alg_2(S[c,d],P)\in \Out_2)
       \\& \leq e^{\epsilon}\cdot\Pr(\alg_1(S'[a,b],P)\in\Out_1)\cdot\Pr(\alg_2(S'[c,d],P)\in \Out_2)\\
       &=e^{\epsilon}\cdot\Pr(\alg_3(S',P)\in\Out)
    \end{align*}
    since $\alg_1$ is $\epsilon$-differentially private and $S[c,d]=S'[c,d]$. The argument for when $i\in [c,d]$ is symmetric. If $i\notin [a,b]\cup[c,d]$, then the output distributions of $S$ and $S'$ are equal.
\end{proof}

%\subsection{$k$-mismatches and Periodicity}
%We will later use the following Lemma to differentiate between two main cases, depending on whether $P$ is close to a periodic string or not.

\section{Upper bounds}\label{sec:upperbounds}
In this section we present our differentially private algorithms for the existence, counting and reporting variants of the approximate pattern matching problem.
\subsection{The Sparse Vector Technique for Approximate Pattern Matching}

\begin{algorithm}[t]
\SetAlgoLined
\DontPrintSemicolon \setcounter{AlgoLine}{0}
\caption{BelowThresh for Approximate Pattern Matching}
\label{alg:BT}
\KwInput{string $S$, pattern $P$, threshold $\thresh$, privacy parameter $\epsilon$}
\KwOutput{a position in $S$ or $\infty$}
$m\leftarrow |P|$\;
$\widetilde{\thresh}=\thresh+\Lap(2/\epsilon)$\;
\For{$i \in [0,|S|-m]$}{
    $d_i=\disth(S[i,i+m-1],P)$\;
    $\tilde{d_i}=d_i+\Lap(4/\epsilon)$\;
    \If{$\tilde{d_i}\leq \widetilde{\thresh}$}{
        \textbf{output} $i$\;
        \textbf{terminate}\;
    }
}
\textbf{output} $\infty$\;
\end{algorithm}
Let $\Lap(b)$ denote a random variable drawn from the Laplace distribution with mean $0$ and scale $b$ as given in Definition \ref{def:laplace}. Note that Fact~\ref{lem:Laplacemech} gives a simple algorithm to compute the Hamming distance between $S[i,i+m-1]$ and $P$, for any fixed $i$: Since the sensitivity of $\disth(S[i,i+m-1],P)$ is 1, we can add Laplace noise scaled with $1/\epsilon$, and this gives an additive error of $O(\ln(1/\beta)/\epsilon)$ with probability $1-\beta$~\cite{DBLP:journals/fttcs/DworkR14}. However, if we would apply the Laplace mechanism to compute $\disth(S[i,i+m-1],P)$ for \emph{all} $i\in [0,n-m]$, then, since changing one position in $S$ changes up to $m$ of the values of $\disth(S[i,i+m-1],P)$, the sensitivity is $m$. This results in an additive error of $O((m/\epsilon)\ln(1/\beta))$ with probability $1-\beta$. Thus, the Laplace mechanism directly applied to this problem is no better than the trivial algorithm of outputting all length-$m$ substrings. Instead, we use a variant of the sparse vector technique (based on an algorithm in \cite{DBLP:conf/stoc/DworkNRRV09} and formally described in \cite{DBLP:journals/fttcs/DworkR14}), which allows to decide for many queries of sensitivity 1 whether the output is above (or in our case, below) a certain threshold, with an error \emph{logarithmic} in the number of queries. Our algorithm for the existence version of the approximate pattern matching problem is given in Algorithm \ref{alg:BT}. The following two facts follow immediately from \cite{DBLP:journals/fttcs/DworkR14}, chapter 3.6:
\begin{lemma}\label{lem:sv_privacy}
    Algorithm~\ref{alg:BT} is $\epsilon$-differentially private.
\end{lemma}
\begin{lemma}\label{lem:sv_acc}
    The output of Algorithm~\ref{alg:BT} fulfills the following properties with probability $1-\beta$ and $\alpha=8\epsilon^{-1}(\ln(|S|-|P|+1)+\ln(2/\beta))$:
    \begin{enumerate}
        \item If Algorithm~\ref{alg:BT} outputs an index $i$, then $\disth(S[i,i+m-1],P)\leq \thresh+\alpha$,
        \item If $i$ satisfies $\disth(S[i,i+m-1],P)\leq \thresh-\alpha$ and Algorithm~\ref{alg:BT} does not terminate before round $i$, then it outputs $i$ and terminates.
    \end{enumerate}
    \end{lemma}

\begin{corollary}\label{cor:existence}
There exists an $\epsilon$-dp algorithm solving the existence variant of $k$-approximate pattern matching with one-sided additive error $\alpha=16\epsilon^{-1}(\ln(|S|-|P|+1)+\ln(2/\beta))$ with probability $1-\beta$.
\end{corollary}
\begin{proof}
    Run Algorithm \ref{alg:BT} with $\thresh=k+8\epsilon^{-1}(\ln(|S|-|P|+1)+\ln(2/\beta))$.
\end{proof}
%(constant can be improved using exponential mechanism (to 8, I think) but not so interesting for the rest of this)

%\subsection{Existence}
%\begin{itemize}
%\item For existence, we first show how to solve the following problem: Let $\tau=\min_{i\in[0,n-1]}\disth(P,S[i,i+m-1])$. Output $j$ satisfying $\disth(P,S[j,j+m-1])\leq \tau+\alpha_1$ with probability $1-\beta/2$.
%\item Using the exponential mechanism (actual citation \cite{DBLP:conf/cpm/GourdelKRS20}, chapter 3.4 in \cite{DBLP:journals/fttcs/DworkR14}), there is an $\epsilon/2$-differentially private algorithm solving this problem with $\alpha_1\leq 4\epsilon^{-1}(\ln n+\ln(2/\beta))$, because the size of the output space (i.e., positions in the string) is $n$ and the sensitivity of the cost function (i.e., cost$(i)=\disth(P,S[i,i+m-1])$) is 1
%\item To get a solution to approximate pattern matching with one-sided error (existence) for $k$, we compute $x=\disth(P,S[j,j+m-1])+\Lap(2/\epsilon)$. This is also $\epsilon/2$-differentially private and gives error at most $\alpha_2=2\epsilon^{-1}\ln(2/\beta)$ with probability $1-\beta/2$. We answer YES if $x\leq k+(\alpha_1+\alpha_2)$. Then it fulfills the requirements of the existence problem for $\alpha=2(\alpha_1+\alpha_2)$
%\item This is almost optimal by lower bound
%\end{itemize}

\subsection{Counting and Reporting}

We will distinguish between different cases, depending on whether $P$ is close to a periodic string with a small period or not. 
We use the following Lemma by Charalampopoulos~et~al.~\cite{conf/focs/Charalampopoulos20}:
\begin{lemma}[Theorem III.1 in \cite{conf/focs/Charalampopoulos20}]\label{lem:twocases}
Given a pattern $P$ of length $m$, a string $S$ of length $n$, and a threshold $k\in[1,\dots,m]$, at least one of the following hold:
\begin{enumerate}
\item  \label{item:twocases_np} The number of $k$-mismatch occurrences is bounded by $576\cdot n/m \cdot k$.\label{case:non-periodic}
\item \label{item:twocases_p} There exists a (primitive) string $Q$ of length $|Q|\leq \frac{m}{128k}$ that satisfies $\disth(Q^{\infty}[0,m-1],P)\leq 2k$.\label{case:periodic}
\end{enumerate}
\end{lemma}
Note that in our privacy definition, only $S$ needs to be private, so we can compute whether case~\ref{item:twocases_p} holds for $P$ without losing any privacy. An example of an algorithm computing this is given in Lemma~\ref{lem:runtime} in Appendix \ref{sec:runtime}. First, we will consider the case where the pattern $P$ is close to a periodic string with small period, and show that in that case, there is a solution to the reporting problem achieving constant multiplicative error and asymptotically optimal additive error. We will call the different cases the ``periodic" and the ``non-periodic" case - note that this is not entirely accurate, since the condition says that $P$ is \emph{close} to a periodic string with \emph{small} period. Thus, $P$ can be aperiodic in the periodic case, and $P$ can be periodic, but with a large period, in the non-periodic case.% We will differentiate the cases ....... Note that this is not fully accurate, since $P$ can actually be periodic and still fulfill \ref{item:twocases_p} in Lemma \ref{lem:twocases}.\textcolor{red}{reformulate to take the in-between case into account!}
\subsubsection{The periodic case}
First, we consider the case where a stronger version of condition \ref{item:twocases_p} in Lemma~\ref{lem:twocases} is true for pattern $P$. In this case we show how to solve the reporting version of the approximate pattern matching problem with constant multiplicative and asymptotically optimal additive error, while satisfying $\epsilon$-differential privacy. We need the following result by Charalampopoulos~et~al.~\cite{conf/focs/Charalampopoulos20}:
%\begin{itemize}
%\item We divide the string into $3n/m$ overlapping blocks of length $\frac{3m}{2}$ (overlapping by $m/2$)
%\item We run an $\epsilon/3$-differentially private algorithm on every such block
%\end{itemize}

\begin{lemma}[Theorem I.7 in \cite{conf/focs/Charalampopoulos20}]\label{lem:per}
Let $P$ denote a pattern of length $m$, let $T$ denote a text of length $n\leq \frac{3m}{2}$, and let $K\in[0,\dots,m]$ denote a threshold. Suppose that both $T[0,m-1]$ and $T[n-m,n-1]$ are $K$-mismatch occurrences of $P$. If there is a positive integer $d\geq 2K$ and a primitive string $Q$ with $|Q|\leq m/(8d)$ and $\disth(P,Q^{\infty}[0,m-1])\leq d$, then each of the following holds:
\begin{enumerate}
\item The string $T$ satisfies $\disth(T,Q^{\infty}[0,n-1])\leq 3d$.
\item Every $K$-mismatch occurrence of $P$ in $T$ starts at a position that is a multiple of $|Q|$.
\item The set of all $K$-mismatch occurrences of $P$ in $T$ can be decomposed into $O(d^2)$ arithmetic progressions with difference $|Q|$.
\end{enumerate}
\end{lemma}

The main idea of our algorithm is now the following: first, we divide $S$ into substrings of length at most $\frac{3m}{2}$. Then for each such substring $T$, we run two instances of Algorithm \ref{alg:BT}, one for $T$ and $P$, and one for their reverse strings. If both instances output an occurrence, then with good probability, a substring of $T$ fulfills the conditions of Lemma~\ref{lem:per} for a suitable value of $K\geq k$, and we can use the Lemma to report all occurrences of distance at most $K$. Else, we know by the properties of Algorithm \ref{alg:BT} that with good probability, there are no occurrences of distance at most $k$ in $T$. The details are given in the proof of the following theorem:
\begin{theorem}\label{thm:periodic}
Let $P$ be a pattern of length $m$. Assume that there exists a primitive string $Q$ of length $|Q|\leq \frac{m}{32C}$ with $C=\max(k, \frac{96(\ln n + \ln (6/\beta))}{\epsilon})$ 
that satisfies $\disth(P,Q^{\infty}[0,m-1])\leq 2k$. %Further, assume $k\geq \frac{48(\ln n + \ln (12/\beta))}{\epsilon}$. 
Then there exists an $\epsilon-$differentially private algorithm for the reporting version of the $k$-approximate pattern matching problem, that given a string $S$ of length $n\geq m$ outputs a set $I\subseteq[0,n-m]$ such that with probability $1-\beta$ the following two conditions are fulfilled: 
\begin{enumerate}
    \item If $\disth(P,S[i,i+m-1])\leq k$, then $i\in I$;
    \item If $i\in I$, then $\disth(P,S[i,i+m-1])\leq (1+\gamma)k+\alpha,$%24k$.
\end{enumerate}
where $\gamma=7$ and $\alpha=6\cdot  \frac{96(\ln n + \ln (6/\beta))}{\epsilon}$.
\end{theorem}

\begin{proof}
First, we compute a $Q$ satisfying the condition above. Note that we can do unlimited computation on $P$ without violating privacy. An algorithm for computing $Q$ is given in Lemma~\ref{lem:runtime} in Appendix~\ref{sec:runtime}. %\textcolor{red}{It would still be nice to discuss how to do this as efficiently as possible: check lemma 4.4 in full version of approximate pattern matching a unified approach paper!}. 
Then,  we divide the string $S$ into overlapping strings of length at most $\lfloor \frac{3m}{2}\rfloor-1=(m-1)+\lfloor \frac{m}{2}\rfloor$. 
We define $\mathcal{F}=\{S[j\cdot \lfloor m/2 \rfloor, j\cdot\lfloor m/2\rfloor +\lfloor 3m/2\rfloor -2], 0\leq j \leq \lfloor\frac{n-m}{\lfloor m/2\rfloor}\rfloor -1\}\cup [\lfloor\frac{n-m}{\lfloor m/2\rfloor}\rfloor \lfloor m/2 \rfloor, n-1]$. Note that any two strings in $\mathcal{F}$ overlap by at most $m-1$ and $\mathcal{F}$ covers $[0,n-1]$. Thus, any occurrence of $P$ in $S$ is included in exactly one string $T\in\mathcal{F}$. Further, any position in $S$ is in at most 3 strings in $\mathcal{F}$, and $|\mathcal{F}|\leq n/\lfloor m/2\rfloor\leq 3n/m$.  For every string $T=S[a,b]\in \mathcal{F}$, we run Algorithm \ref{alg:periodic} and return all positions in $a+I$, where $I$ is the set returned by Algorithm \ref{alg:periodic} on inputs $(T,P,|Q|,k,n,m,\epsilon)$. 

 \begin{algorithm}[t]
\SetAlgoLined
\DontPrintSemicolon \setcounter{AlgoLine}{0}
\caption{Reporting Approximate Pattern Matching, periodic case}
\label{alg:periodic}
\KwInput{string $T$, pattern $P$, $|Q|$, $k$, $n$, $m$, $\epsilon$}
\KwOutput{a set $I$ of positions in $T$}
$\thresh=k+\epsilon^{-1}48(\ln (m/2) + \ln(12(n/m)/\beta))$\;
$\epsilon'=\epsilon/6$\;
$i\leftarrow$ output of Algorithm \ref{alg:BT} on input (string $T$, pattern $P$, threshold $\thresh$, privacy parameter $\epsilon'$)\;
$j'\leftarrow$ output of Algorithm \ref{alg:BT} on input (string $\rev{T}$, pattern $\rev{P}$, threshold $\thresh$, privacy parameter $\epsilon'$)\;
    \If{$j'=\infty$ \bf{or} $i=\infty$}{\textbf{output $\emptyset$}\; \textbf{terminate}}
$j=(|T|-1)-j'-(m-1)$ \tcp*{translate starting position in $\rev{T}$ to starting position in $T$}
\bf{output} $I=\{i+\ell|Q|,0\leq \ell \leq \lfloor \frac{j-i}{|Q|} \rfloor\}$

\end{algorithm}

{\bf Privacy analysis.} Note that in every instance of Algorithm~\ref{alg:periodic}, we run two instances of Algorithm~\ref{alg:BT} with privacy parameter $\epsilon/6$. By Lemma~\ref{lem:sv_privacy} and Fact~\ref{fact:composition_theorem}, Algorithm~\ref{alg:periodic} is $\epsilon/3$-differentially private. Further, let $S$ and $S'$ differ in position $i^{*}$. Since $i^{*}$ can only be in at most three strings in $\mathcal{F}$, the full algorithm on $S$ satisfies $\epsilon$-differential privacy by Fact~\ref{fact:composition_theorem} and Lemma~\ref{lem:composing_disjiont}.

{\bf Accuracy analysis.} Fix a $T$ in $\mathcal{F}$. Let $i$ and $j$ be as in Algorithm \ref{alg:periodic} on input $T$. If $j'$ was set to $\infty$, let $j=-\infty$. Let $\beta'=\beta/(6(n/m))$. Note that by Lemma~\ref{lem:sv_acc}, with probability at least $1-\beta'$, we have for all $i'<i$,
\begin{align}\begin{split}
    \label{condlowerboundi} \disth(T[i',i'+m-1],P)&>k +\epsilon^{-1}48(\ln (m/2) + \ln(12(n/m)/\beta)) - \epsilon^{-1}48(\ln(m/2) + \ln(2/\beta'))\\&=k,\end{split}
\end{align}
and, if $i<\infty$,
\begin{align}
\begin{split}
\label{condupperboundi} \disth(T[i,i+m-1],P)&\leq k + \epsilon^{-1}48(\ln (m/2) + \ln(12(n/m)/\beta)) + \epsilon^{-1}48(\ln(m/2) + \ln(2/\beta'))\\
&= k + \epsilon^{-1}96(\ln(6n/\beta)).
\end{split}
\end{align}
Similarly, also with probability $1-\beta'$, we have
for all $i''>j$,
\begin{align}
    \label{condlowerboundj} \disth(T[i'',i''+m-1],P)>k.
\end{align}
and, if $j>-\infty$,
\begin{align}\label{condupperboundj} \disth(T[j,j+m-1],P)\leq k +  \epsilon^{-1}96(\ln(6n/\beta)).
\end{align}

Thus, with probability $1-\beta/(3(n/m))$, both conditions are true, and since $|\mathcal{F}|\leq 3n/m$, these conditions are true with probability at least $1-\beta$ over all instances of Algorithm \ref{alg:periodic}. In the following, we condition on that.

If either $j'$ or $i$ was set to $\infty$, then there is no occurrence of distance at most $k$ in $T$, and in this case we return the empty set. Next, if $j<i$, then there is also no occurrence of at most $k$ in $T$ by (\ref{condlowerboundi}) and (\ref{condlowerboundj}). Note that also in this case, Algorithm \ref{alg:periodic} returns the empty set. 

Now, consider the case $j\geq i$ for finite integers $i$ and $j$. We want to argue that in this case, the string $T[i,j+m-1]$ fulfills the conditions of Lemma~\ref{lem:per} for an appropriate choice of $K>k$. Obviously, $|T[i,j+m-1]|\leq |T| \leq \frac{3m}{2}$. We set $K=k+\epsilon^{-1}96(\ln(6n/\beta))\leq 2C$. By (\ref{condupperboundi}) and (\ref{condupperboundj}) both $i$ and $j$ are the start of a $K$-mismatch occurrence. Let $d=2K\leq 4C$. By assumption, there is a primitive string $Q$ with $|Q|\leq m/(32 C) \leq m/8d$ with $\disth(P,Q^{\infty}[0,m-1])\leq 2k \leq d$. Thus, the conditions of Lemma~\ref{lem:per} are fulfilled.
 This gives the following:
 \begin{enumerate}
     \item Since the string $T[i,j+m-1]$ satisfies $\disth(T[i,j+m-1],Q^{\infty}[0,j+m-i-1])\leq 3d$, we have that for any position $q=i+\ell|Q|$  for $\ell\in [0,\lfloor \frac{j-i}{|Q|}\rfloor]$:
    \begin{align*}
       % &\disth(T[i+\ell|Q|,i+\ell|Q|+m-1,P])\\&\leq \disth(T[i+\ell|Q|,i+\ell|Q|+m-1],Q^{\infty}[0,m-1])+ \disth(Q^{\infty}[0,m-1],P) \\&\leq 3d+2k=8k+6\cdot 96 \cdot \epsilon^{-1}(\ln(4n/\beta).
        \disth(T[q,q+m-1],P)&\leq \disth(T[q,q+m-1],Q^{\infty}[0,m-1])+ \disth(Q^{\infty}[0,m-1],P) \\&\leq 3d+2k=8k+6\cdot 96 \cdot \epsilon^{-1}(\ln(6n/\beta)).
    \end{align*}
    Thus, every reported occurrence $q$ fulfills $\disth(T[q,q+m-1],P)\leq(1+\gamma)k+\alpha$ with $\gamma=7$ and $\alpha=6\cdot 96 \cdot \epsilon^{-1}(\ln(6n/\beta))$.
    \item Since every $K$-mismatch occurrence of $P$ in $T[i,j+m-1]$ starts at a multiple of $|Q|$, then in particular, any $k$-mismatch occurrence of $P$ in $T[i,j+m-1]$ starts at a position $i+\ell|Q|$ in $T$ for $\ell\in [0,\lfloor \frac{j-i}{|Q|}\rfloor]$. Thus, any substring of $T[i,j+m-1]$ of length $m$ that does not start at $i+\ell|Q|$  for some $\ell\in [0,\lfloor \frac{j-i}{|Q|}\rfloor]$ has a distance larger than $k$. 
    \end{enumerate}
    Further, by (\ref{condlowerboundi}) and (\ref{condlowerboundj}), $\disth(T[i',i'+m-1],P)>k$ for all $i'<i$ or $i'>j$. Thus, we report all occurrences with distance at most $k$.
\end{proof}

\subsubsection{The non-periodic case}

%    \item \textcolor{red}{ideally try to get a similar statement as the periodic case}
    Next, we assume condition \ref{item:twocases_p} in Lemma~\ref{lem:twocases} is not true for $P$, that is, there does not exist a string $Q$ of length $|Q|\leq \frac{m}{128k}$ that satisfies $\disth(Q^{\infty}[0,m-1],P)\leq 2k$. This means the number of $k$-mismatch occurrences in any string $T$ of length $|T|$ is bounded by $576\cdot |T|/m \cdot k$ by Lemma \ref{lem:twocases}.
    In particular, in any substring of length $\leq 2m$ of $S$, the number of occurrences is at most $1152k=O(k)$. We will use this fact to solve the counting variant of the problem in the non-periodic case.
%    \textcolor{red}{I did not define $Q^*$ fix this}
    Note that Theorem~\ref{thm:periodic} and Theorem~\ref{thm:non-periodic} do not cover all the cases: If $k\leq C/4$, where $C$ is as in Theorem~\ref{thm:periodic}, then it is possible that the conditions of neither theorem are fulfilled. We deal with that case later.
%    \item \textcolor{red}{use Fact 3.1 in \cite{DBLP:conf/soda/CliffordFPSS16} instead?}
    
   % \item This means that the sensitivity for counting all occurrences of distance at most $k$ is bounded by $1152k=O(k)$, thus adding Laplace noise scaled with that value satisfies $\epsilon$-differential privacy; but that doesn't give us any useful guarantees
   % \item \textcolor{red}{However we defined error in a different way, let's do both here and then see what makes sense in the periodic setting}
   \begin{theorem}\label{thm:non-periodic}
   Let $P$ be a pattern of length $m$. If there does not exist a string $Q$ of length $|Q|\leq \frac{m}{128k}$ that satisfies $\disth(Q^{\infty}[0,m-1],P)\leq 2k$, then there exists an $\epsilon-$differentially private algorithm that given a string $S$ of length $n\geq m$ computes a count $c$, such that with probability $1-\beta$ it holds that
   $c_k(S)\leq c \leq c_{(1+\gamma)k}(S)$, where $\gamma=O(\epsilon^{-1}\cdot(\ln n + \ln(1/\beta))$. Further, if $c>0$, it returns a witness $i$ satisfying $\disth(P,S[i, i+m-1])\leq (1+\gamma)k$.%$c_k\leq c \leq c_{k+\alpha}$, where $\alpha=O(k\epsilon^{-1}\cdot(\ln n + \ln(k/\beta))$. Further, it returns a witness $i$ satisfying $\disth(P,S[i\dots i+m-1])\leq k+\alpha$.
   \end{theorem}
   \begin{proof}
   The first step is to divide the string $S$ into substrings of length at most $2m-1$, which form overlapping blocks, such that any pattern occurrence appears in exactly one block. That is, we define the set $\mathcal{B}=\{S[jm,(j+2)m-2],j=0\dots\lfloor \frac{n+1}{m} \rfloor-2\}\cup \{S[(\lfloor\frac{n+1}{m} \rfloor-1)m,n-1]\}$. Since $\mathcal{B}$ covers $[0,n-1]$ and two strings overlap by at most $m-1$, any pattern occurrence in $S$ is contained in exactly one string in $\mathcal{B}$. Note that any position in $S$ is included in at most two strings in $\mathcal{B}$.

\begin{algorithm}[t]
\SetAlgoLined
\DontPrintSemicolon \setcounter{AlgoLine}{0}
\caption{Counting Approximate Pattern Matching, non-periodic case}
\label{alg:non-periodic}
\KwInput{string $T$, pattern $P$, $k$, $n$, $m$, $\epsilon$}
\KwOutput{a count $c$ and a position $j$ in $T$}
$j=-1$\;
$i=-1$\;
$c=0$\;
$\thresh=k+\epsilon^{-1}16\cdot 1152 k(\ln m + \ln(2(n/m)1152 k/\beta))$\;
$\epsilon'=\epsilon/(2\cdot 1152k)$\;
\While{$i< |T|-|P|\And c< 1152 k$}{
    $j\leftarrow$ output of Algorithm \ref{alg:BT} on input (string $T[i+1,n-1]$, pattern $P$, threshold $\thresh$, privacy parameter $\epsilon'$)\;
    \If{$j=\infty$}{\textbf{output $(c,i)$}\; \textbf{terminate}}
    $c=c+1$\;
    $i=j$}
    \textbf{output} $(c, j)$
\end{algorithm}

For each $T\in \mathcal{B}$, we run Algorithm~\ref{alg:non-periodic}. Then for the outputs $(c_1,j_1),\dots,(c_{|\mathcal{B}|}, j_{|\mathcal{B}|})$, we output $\sum_{\ell=1}^{|\mathcal{B}|} c_{\ell}$. If there exists a $j_{\ell}>-1$, we choose an arbitrary such and output $\ell m+j_{\ell}$.

\paragraph*{Privacy analysis.} For any instance of Algorithm~\ref{alg:non-periodic}, we run at most $1152 k$ instances of Algorithm~\ref{alg:BT} with privacy parameter $\epsilon'=\epsilon/(2\cdot 1152 k)$. Thus any instance of Algorithm~\ref{alg:non-periodic} is $\epsilon/2$-differentially private by Lemma~\ref{lem:sv_privacy} and Fact~\ref{fact:composition_theorem}. Further, let $S$ and $S'$ differ in position $i^{*}$. Since $i^{*}$ can only be in at most two strings in $\mathcal{B}$, the full algorithm satisfies $\epsilon$-differential privacy by Fact~\ref{fact:composition_theorem} and Lemma~\ref{lem:composing_disjiont}.

\paragraph*{Accuracy analysis.}

    %\item divide the string into $\leq n/m+1$ overlapping blocks of length $2m-1$, overlapping by $m-1$ characters, such that any pattern occurrence is included in exactly one block and any position in $S$ is included in at most 2 blocks. 
   % \item For each of these blocks, run Algorithm~\ref{alg:BT} with privacy parameter $\epsilon/(2q\cdot k)$, threshold $k+\alpha'$, where $q=1152$ and $\alpha'=\epsilon^{-1}16q\cdot k (\ln m + \ln(2(n/m)q\cdot k)/\beta)$. When we terminate at a position $i$, we run the algorithm again starting from $i+1$, until either
 %   \begin{itemize}
 %       \item we terminated $q\cdot k$ times or
 %       \item we reached the end of the current block.
 %   \end{itemize}
 %   We return the number of times we terminated for each block and return the sum for all blocks.
 %   \item Privacy: For each block, we run at most $q\cdot k$ instances of Algorithm~\ref{alg:BT}. Since we use privacy parameter $\epsilon/(2q\cdot k)$, their composition is $\epsilon/2$-dp. Any position can be in at most two blocks, so the composition of all instantiations of Algorithm~\ref{alg:BT} is $\epsilon$-dp.
    %If the sparse vector technique terminates before we processed the entire block, we return $1152k$.
Let $c(T)$ be the output of Algorithm~\ref{alg:non-periodic} for string $T\in\mathcal{B}$ and $c_{k}(T)$ the true count of positions $i$ such that $\disth(T[i,i+m-1],P)\leq k$. For a fixed $T$, we will show that $c_{k}(T)\leq c(T) \leq c_{(1+\gamma)k}(T)$ with probability $1-\beta/(n/m)$. Since $|\mathcal{B}|\leq n/m$, a union bound then implies that the bound holds for all $T\in\mathcal{B}$ with probability $1-\beta$. Note that since any substring of length $m$ of $S$ is included in exactly one string in $\mathcal{B}$, this implies $c_k(S)=\sum_{T\in\mathcal{B}}c_k(T)\leq \sum_{T\in\mathcal{B}}c(T)\leq \sum_{T\in\mathcal{B}}c_{(1+\gamma)k}(T)=c_{(1+\gamma)k}(S)$. \\
    
    Now, fix $T\in\mathcal{B}$ and let $\alpha'=8(\epsilon')^{-1}(\ln(|T|-|P|+1)+\ln(2/\beta'))$. By Lemma~\ref{lem:sv_acc}, with probability at least $1-\beta'$, whenever an instance of Algorithm~\ref{alg:BT} in Algorithm~\ref{alg:non-periodic} returns a position $i$, the distance $\disth(T[i,i+m-1],P)\leq \thresh +\alpha'$; further, any position $i'\leq i$ which was part of that instance satisfies $\disth(T[i',i'+m-1],P)>\thresh-\alpha'$ (otherwise it would have been output instead of~$i$). Thus, for each such $i$ and $\beta'=\beta/((n/m)1152 k)$ we have
    \begin{align*}
        \disth(T[i,i+m-1],P)\leq k+&\epsilon^{-1}16\cdot 1152 k(\ln m + \ln(2(n/m)1152 k/\beta)) +\alpha'\\
        =k+&\epsilon^{-1}16\cdot 1152 k(\ln m + \ln(2(n/m)1152 k/\beta)) \\+8(&\epsilon')^{-1}(\ln(|T|-|P|+1)+\ln(2/\beta'))\\
        =k+&\epsilon^{-1}16\cdot 1152 k(\ln m + \ln(2(n/m)1152 k/\beta)) \\+&\epsilon^{-1}16\cdot 1152 k(\ln m+\ln(2/\beta'))\\
        =k+&2\alpha',
    \end{align*}
    and for each $i'\leq i$ in that instance of Algorithm \ref{alg:BT}
        \begin{align*}
        \disth(T[i',i'+m-1],P)&> k+\epsilon^{-1}16\cdot 1152 k(\ln m + \ln(2(n/m)1152 k/\beta)) -\alpha'\\
        &=k,
    \end{align*}
    with probability $1-\beta'$. Thus, over the entire run of Algorithm \ref{alg:non-periodic}, the inequalities hold with probability at least $1-\beta/(n/m)$, and we condition on that. It directly follows that all counted positions $i$ satisfy $\disth(T[i,i+m-1],P)\leq k+2\alpha' = (1+\gamma) k$, for $\gamma=\epsilon^{-1}32\cdot 1152 (\ln m + \ln(2(n/m)1152k/\beta)=O(\epsilon^{-1}(\ln n + \ln(1/\beta)))$. Thus, $c(T) \leq c_{(1+\gamma)k}(T)$. For the lower bound, there are two cases to consider: \\
    {\bf Case 1}: If $c< 1152k$ when Algorithm \ref{alg:non-periodic} ends, then every possible starting position $i\leq |T|-|P|$ was considered by some instance of Algorithm \ref{alg:BT}. Thus, all positions $i$ satisfying $\disth(T[i,i+m-1],P)\leq k$ were counted and $c_{k}(T)\leq c(T) \leq c_{(1+\gamma)k}(T)$.\\
    {\bf Case 2}: If $c=1152k$, then $c_k(T)\leq 1152 k$ holds by Lemma~\ref{lem:twocases} and since $|T|<2m$.
  %  By the union bound, this holds for all instances of Algorithm \ref{alg:BT} called in one instance of Algorithm \ref{alg:non-periodic} with probability $1-\beta/(n/m)$. and whenever we answer don't return a position, the distance is more than $k$, with probability at least $1-\beta/((n/m)kq)$. Since for each block we run at most $kq$, the guarantee holds for the entire block with probability $1-\beta/(n/m)$. Thus, if we process the entire block without terminating, we return $c$ satisfying $c^B_k\leq c^B \leq c^B_{k+\alpha}$ for $\alpha=2\alpha'$. 
   % \item If we terminate preemptively, we output $c^B=1152k\geq c_k$ by Lemma \ref{lem:twocases}. Further, we still have that for each returned position $i$, we had that the distance was at most $k+2\alpha'$, so we still have $c^B\leq c^B_{k+\alpha}$ for $\alpha=2\alpha'$, again with probability $1-\beta/(n/m)$.
   % \item As a witness we return any position returned by any of the instances of Algorithm~\ref{alg:BT}.
   % \item Since no two blocks share a substring of length $m$, we don't count double. The accuracy guarantees hold simulatenously for all $(n/m)$ blocks with probability $1-\beta$.
    %\end{itemize}
    \end{proof}

\subsubsection{\texorpdfstring{Non-periodic and small $k$}{Non-periodic and small k}}
Note that there can be a case where neither the conditions of Theorem~\ref{thm:non-periodic} nor Theorem~\ref{thm:periodic} are fulfilled: If $k<C/4=24\epsilon^{-1}\ln(6n/\beta)$, and there exists a primitive string $Q$ of length $|Q|\leq m/(128k)$ such that $\disth(P,Q^{\infty}[0,m-1])\leq 2k$, but there does not exist a primitive string $Q'$ of length $|Q'|\leq m/(32C)$ such that $\disth(P,Q'^{\infty}[0,m-1])\leq 2k$. Note that the second condition implies that there does not exist a primitive string $Q'$ of length $|Q'|\leq m/(128K)$ such that $\disth(P,Q'^{\infty}[0,m-1])\leq 2k< 2K$, for $K=C/4$.
\begin{lemma}\label{lem:smallk}
   Let $P$ be a pattern of length $m$. If $k<K=24\epsilon^{-1}\ln(6n/\beta)$ and there does not exist a string $Q$ of length $|Q|\leq m/(128K)$ such that $\disth(P,Q^{\infty}[0,m-1])\leq 2K$, then there exists an $\epsilon$-differentially private algorithm that given a string $S$ of length $n\geq m$ computes a count $c$, such that with probability $1-\beta$ it holds that $c_k(S)\leq c \leq c_{k+\alpha}(S)$, where $\alpha=O(\epsilon^{-2}(\ln^2(n/\beta)))$.
\end{lemma}
\begin{proof}
    Note that the conditions of Theorem~\ref{thm:non-periodic} are fulfilled with $K$ taking the role of $k$. Thus there exists an algorithm that outputs a count $c$ such that with probability $1-\beta$ it holds that $c_K(S)\leq c \leq c_{(1+\gamma)K}(S)$ where $\gamma=O(\epsilon^{-1}(\ln(n/\beta)))$. The lemma now follows since $c_k(S)\leq c_K(S)$ and $c_{(1+\gamma)K}(S)\leq c_{\eta^2}(S)\leq c_{k+\eta^2}(S)$ for $\eta=\max(1+\gamma,K)=O(\epsilon^{-1}(\ln(n/\beta)))$.
\end{proof}
%\begin{itemize}
%    \item Note that if neither the conditions of Theorem 1 nor 2 are true, then $k\leq 48\epsilon^{-1}(\ln n + \ln(6/\beta))=:q$ and the conditions of Theorem 1 hold for $q$. Thus, we run the algorithm of Theorem 1 with $q$ taking the place of $k$, and get a count $c$ satisfying $c_k(S)\leq c_q(S)\leq c_{(1+\gamma)q}(S)\leq c_{\gamma^2}(S)\leq c_{k+\gamma^2}$, for $\gamma=O(\epsilon^{-1}(\ln(n/\beta))$. Thus, in this case we have an additive error of $\gamma^2=O((\epsilon^{-1}(\ln(n/\beta)))^2)$.
%\end{itemize}
Theorem~\ref{thm:main} now follows by noticing that any pattern $P$ fulfills the conditions of either Theorem~\ref{thm:periodic}, Theorem~\ref{thm:non-periodic} or Lemma \ref{lem:smallk}, and that the reporting solution from Theorem~\ref{thm:periodic} implies a counting solution with the same error bounds.

\section{Lower bound}\label{sec:lowerbound}
For any $k\leq m$, there is a trivial algorithm solving the reporting version of the approximate pattern matching problem with additive one-sided error $O(m-k)$ with probability 1 while preserving $\epsilon$-differential privacy: We just output every position $i\in[0,n-m+1]$.
The next Theorem shows that in order to have error $o(m-k)$, we need $m=\Omega(\ln n)$, and in that case the additive error is $\Omega(\ln(n/m))$. Note that the lower bound holds for any pattern $P$ and for the existence or counting variant, as long as at least one witness is returned. Our lower bound is based on a packing argument.\\

%\textcolor{red}{Can we strengthen that lower bound? How about one-sided error? non-periodic patterns? small $k$?}

\begin{theorem}\label{thm:lowerbound}
    Let $P$ be any string of length $m$ and let $k<m$ be a parameter. Assume there is an $\epsilon$-differentially private algorithm $\alg$ with the following guarantee: If $S$ is a string of length $n\geq m$ such that there exists $j\in[0,n-m]$ with $\disth(S[j,j+m-1],P)\leq k$, then with probability at least $2/3$, $\alg$ returns a position $i\in[0,n-m]$ such that $\disth(S[i,i+m-1],P)\leq k+\alpha$. Then either $\alpha=\Omega(m-k)$, or $m=\Omega(\epsilon^{-1}\ln n)$ and $\alpha=\Omega(\epsilon^{-1}\ln(n/m))$.
\end{theorem}
\begin{proof}
    First, we assume there is an algorithm $\alg$ as in the statement of the theorem satisfying $\alpha<m-k$. We show $m=\Omega(\ln n)$. 
    We start by dividing $[0,n-1]$ into disjoint intervals of length $m$ (we assume wlog that $n$ is a multiple of $m$). That is, we define the set $\mathcal{I}=\{[jm,(j+1)m-1], j=0,\dots, n/m-1\}$. For every \emph{even}  $j\in\{0,\dots, n/m-1\}$, we define a string $S_j$ as follows: $S_j[jm,(j+1)m-1]=P$, and for all $q\in[0,n-1]\backslash[jm,(j+1)m-1]$, we set $S_j[q]=\$$ for some $\$$ which does not appear in $P$. 
    
    Note that $S_j$ and $S_i$ have a Hamming distance of $2m$ for all even $i\neq j$, $i,j\in [0,n/m-1]$. Further, we have $\disth(S_j[jm,(j+1)m-1],P)=0\leq k$, and for every $q\in[0,n-m]\backslash[(j-1)m+1,(j+1)m-1]$, we have $\disth(S_j[q,q+m-1],P)=m> k+\alpha$. Thus, by assumption on $\alg$, we have 
    \begin{align*}
        \Pr(\alg(S_j)\in [(j-1)m+1,(j+1)m-1])\geq 2/3,
    \end{align*} and, by group privacy (Fact \ref{fact:group_privacy}),
    \begin{align*}
    \Pr(\alg(S_j)\in[(i-1)m+1,(i+1)m-1])\geq e^{-2m\epsilon}2/3,
    \end{align*}
    for every even $i\in [0,n/m-1]$.
    Since these events are disjoint, we have
    \begin{align*}
        1\geq \sum_{\textnormal{even }i\in [0,n/m-1]}e^{-2m\epsilon}2/3
    \end{align*}
    and therefore 
    \begin{align*}
        m\geq (2\epsilon)^{-1}(\ln(n/(2m))+\ln(2/3)),
    \end{align*}
    and therefore $m=\Omega(\epsilon^{-1}\ln n)$.

    Next, we want to show $\alpha=\Omega(\ln(n/m))$. For this, we consider the same partition $\mathcal{I}$ into intervals, and for every even $j$ in $[0,n/m-1]$ we define $S_j$ as follows: $S_j[jm,(j+1)m-1]=\$^k P[k,m-1]$, and for every even $i\neq j$, $i\in [0,n/m-1]$, we define $S_j[im,(i+1)m-1]=\$^{k+\alpha+1}P[k+\alpha+1,m-1]$. For all other positions $q\in[0,n-1]$, we define $S_j[q]=\$$. We have $\disth(S_j[jm,(j+1)m-1],P)\leq k$ and $\disth(S_j[q,q+m-1],P)>k+\alpha$ for all $q\in[0,n-m]\backslash[(j-1)m+1,(j+1)m-1]$. Further, all $S_j$, $S_i$ with $i,j$ even and $i\neq j$ have a Hamming distance of $2\alpha+2$. By assumption on $\alg$ we have
    \begin{align*}
        \Pr(\alg(S_j)\in[(j-1)m+1,(j+1)m-1])\geq 2/3,
    \end{align*} and, by group privacy (Fact \ref{fact:group_privacy}),
    \begin{align*}
    \Pr(\alg(S_j)\in[(i-1)m+1,(i+1)m-1])\geq e^{-(2\alpha+2)\epsilon}2/3.
    \end{align*}
    for every even $i\in [0,n/m-1]$.
    Since these events are disjoint, we have
    \begin{align*}
        1\geq \sum_{\textnormal{even }i\in [0,n/m-1]}e^{-(2\alpha+2)\epsilon}2/3
    \end{align*}
    and therefore 
    \begin{align*}
        \alpha\geq (2\epsilon)^{-1}(\ln(n/2m)+\ln(2/3))-1,
    \end{align*}
    and therefore $\alpha=\Omega(\epsilon^{-1}\ln (n/m))$.
\end{proof}

\section{Conclusion}\label{sec:conclusion}

We have initiated a study of \emph{differentially private pattern matching} algorithms, and have shown that combining techniques from the areas of differential privacy and pattern matching can be used to obtain interesting new results. Specifically, for the approximate pattern matching problem with Hamming distance under $\epsilon$-differential privacy, we have both shown a strong lower bound and new upper bounds. The upper bounds asymptotically match the lower bound for the existence variant, and for the reporting variant for a special class of patterns. %Specifically, we have given differentially private algorithms achieving an asymptotically optimal $O(\log n)$ additive error for the existence variant of the $k$-approximate pattern matching problem under Hamming distance; an $O(1)$ multiplicative and $O(\log n)$ additive error for the reporting variant for a special class of patterns;  and a counting algorithm for any pattern with $O(\log n)$ multiplicative error, or $O(\log^2 n)$ additive error in the case where $k$ is small. All these error bounds hold for constant $\epsilon$ and with constant probability.
There are many potential directions for future research, including:
\begin{itemize}
    \item closing the gap between the upper and the lower bound for \emph{all} patterns;
    \item studying $(\epsilon,\delta)$-differential privacy for this problem;
    \item considering other distance measures, e.g. edit distance, both for the definition of $k$-approximate pattern matching, and for the privacy definition;
    \item considering other error measures, e.g. for the counting variant of pattern matching.
\end{itemize}
Further, it would be exciting to see if it is possible to obtain differentially private indexing data structures with useful error guarantees.

\section{Acknowledgements}{This work was supported by a research grant (VIL51463) from VILLUM FONDEN.}

\bibliographystyle{plain}
\bibliography{references}

\appendix
\section{Runtime Analysis}\label{sec:runtime}

In the following, we analyze the runtime of our algorithms and show that it is $O(nm+m^3)$, assuming that noises from the Laplace distribution can be drawn in constant time. We note that in this work we did not optimize for runtime.

First, note that computing the Hamming distance between $S[i,i+m-1]$ and $P$ for any $i$ can be done in $m$ time. We collect some immediate observations about the runtimes of the given algorithms, if we already know whether $P$ fulfills the conditions of the theorems (and for which $|Q|$).
\begin{fact}
    Let $i$ be the output of Algorithm \ref{alg:BT} on an input string $T$ and pattern $P$. The runtime of Algorithm \ref{alg:BT} is $O(\min(i\cdot m, |T|\cdot m))$.
\end{fact}
\begin{corollary}
    The runtime of Algorithm \ref{alg:periodic} on input string $T$ and pattern $P$ is $O(|T|\cdot m)$.
\end{corollary}
\begin{corollary}
    The runtime of Algorithm \ref{alg:non-periodic} on input string $T$ and pattern $P$ is $O(|T|\cdot m)$.
\end{corollary}
\begin{corollary}
    Given $P$ and $|Q|$ satisfying the conditions of Theorem \ref{thm:periodic}, the algorithm given by Theorem \ref{thm:periodic} has a runtime of $O(nm)$.
\end{corollary}
\begin{corollary}
    The algorithms of Theorem \ref{thm:non-periodic} and Lemma \ref{lem:smallk} have a runtime of $O(nm)$.
\end{corollary}

Next, we analyze the ``preprocessing" part for $P$, i.e. we show how to decide if $P$ is close to a periodic string $Q^{\infty}$ with small $|Q|$. 

\begin{lemma}\label{lem:runtime}
    Let $P$ be a pattern of length $m$ and let $k$ be a parameter. In $O(m^3)$ time, we can decide if there exists a $Q$ such that $|Q|\leq \max(\frac{m}{32C},\frac{m}{128k})$ fulfilling $\disth(P,Q^{\infty}[0,m-1])\leq 2k$, where $C$ is defined as in Theorem~\ref{thm:periodic}, and compute the shortest such. 
\end{lemma}
\begin{proof}
    For any potential $q\leq \max(\frac{m}{32C},\frac{m}{128k})$, we do the following: First, we conceptually partition the pattern $P$ into blocks of length $q$. Note that there are at least $m/q\geq \min(128 k,32C)\geq 32k$ such blocks. Now assume there exists $Q$ of length $|Q|=q$ satisfying $\disth(P,Q^{\infty}[0,m-1])\leq 2k$. Then, since $\disth(P,Q^{\infty}[0,m-1])\leq 2k$, \emph{all but at most $2k$ blocks} of $P$ have to be equal to $Q$. Note that there can be at most one potential string of length $q$ fulfilling that condition. To find it, we traverse $P$ and count how often a block in $P$ is equal to any given substring of length $q$. We can do this by e.g. building a trie of all blocks as we traverse $P$. This takes $O(m)$ time. Now, if we found a candidate string $Q$ such that all but at most $2k$ blocks are equal to $Q$, we spend at most $m$ time to check if indeed $\disth(P,Q^{\infty}[0,m-1])\leq 2k$. Since there are at most $\max(\frac{m}{32C},\frac{m}{128k})\leq m$ possible values of $q$, the total runtime is $O(m^3)$.
\end{proof}
Note that the condition of Lemma~\ref{lem:smallk} can be checked by applying Lemma~\ref{lem:runtime} with $C/4$ taking the role of $k$.
\end{document}